%% file: bpuc_ejor.tex
\documentclass[preprint, 3p, 12pt]{elsarticle}
\usepackage{times}
\usepackage{amsmath}
\usepackage{amssymb}
\usepackage{graphicx}
\usepackage{color}
\usepackage[english]{babel}
\usepackage{lscape}
\usepackage{graphicx,color}
\usepackage{txfonts}
\usepackage{subfigure}
\usepackage{wrapfig}
\usepackage{dsfont}
\usepackage{algorithm}
\usepackage{algorithmic}
\usepackage{tikz}
\usepackage{url}
\usepackage{amssymb}
\usepackage{rotating}

\newtheorem{example}{Example}

\newtheorem{proposition}{Proposition}
\newtheorem{proof}{Proof}

\newenvironment{coden}
	{\begin{tt}\begin{tabbing}12345\=12\=12\=12\=12\=12\=\kill}
	{\end{tabbing}\end{tt}}

\newcommand{\caplab}[2]{\caption{\label{#1} #2}}
\newcommand{\tabbegin}[1]{\begin{table*}[#1]\centering}
\newcommand{\tabend}{\end{table*}}
\newcommand{\figbegin}[1]{\begin{figure*}[#1]\centering}
\newcommand{\figend}[2]{\caplab{#1}{#2}\end{figure*}}

\newcommand\arcflow{\textsc{Arc-Flow}}

\begin{document}

\title{Bin Packing with Linear Usage Costs
}
\author[rvt1]{Hadrien Cambazard}
\ead{hadrien.cambazard@grenoble-inp.fr}
\author[rvt2]{Deepak Mehta\corref{cor1}}
\ead{deepak.mehta@insight-centre.org}
\author[rvt2]{Barry O'Sullivan}
\ead{barry.osullivan@insight-centre.org}
\author[rvt2]{Helmut Simonis}
\ead{helmut.simonis@insight-centre.org}


\address[rvt1]{G-SCOP, Universit\'e de Grenoble; Grenoble INP; UJF Grenoble 1; CNRS, France}
\address[rvt2]{Insight Centre for Data Analytics, Universiy College Cork, Ireland}


\begin{abstract}
Bin packing is a well studied problem involved in many applications.
The classical bin packing problem is about minimising the number of bins
and ignores how the bins are utilised. 
We focus in this paper, on a variant of bin packing that is at the heart of efficient management of data centres. 
In this context, servers can be viewed as bins and virtual machines as items. 
The efficient management of a data-centre involves minimising energy costs while ensuring service quality.
The assignment of virtual machines on servers and how these servers are utilised has a huge impact on the energy consumption. 
 We focus on a  bin packing problem where linear costs are associated to the use of bins to model the energy consumption. We study lower bounds based on Linear Programming 
and extend the bin packing global constraint with  cost information. 
\end{abstract}

 \begin{keyword}
 Bin Packing, Constraint Programming
 \end{keyword}

\maketitle
\section{Introduction}
\label{sec-introduction}

\input{introduction}

\section{Bin Packing with Linear Usage Costs}
\label{bplincost}
\input{bpuc}

\section{Experimental Results}
\label{sec-experiments}
In this section we  report the results for solving  bin packing with usage cost problem to show 
the efficiency of the BPUC global constraint.

We  compare on randomly generated instances the lower bounds $z^*_1, z^*_2, z^*_3$ as well as exact algorithms: Model (1), \arcflow\  Model~(3), and two CP models using the \textsc{BinPackingUsageCost} constraint. The second CP model referred to as CP+CG activates the propagation of  the $z_2^*$ bound by the \textsc{BinPackingUsageCost} constraint.  Standard symmetry/dominance breaking techniques for BP are applied to the MIP~\cite{orbital2012} of Model (1) and CP~\cite{DBLP:conf/cp/Shaw04}. A random instance  is defined by  $(n,m,X)$, where  
 $n$ is the number of items ($n \in \{15, 25, 200, 250, 500\}$),  $m$ is the number of bins ($m \in  \{10, 15, 25, 30\}$), and  parameter $X \in \{1,2,3\}$ denotes that the item sizes are uniformly randomly generated in the intervals $[1,100]$, $[20,100]$, and $[50,100]$ respectively. 
The  capacities of the bins are picked randomly from the sets $\{80,$ $100,$ $120,$ $150,$ $200,$ $250\}$ and
$\{800, 1000, 1200, 1500,$ $2000,$ $2500\}$ 
when $n \in \{15, 25\}$ and  $n \in \{200, 250, 500\}$ respectively. The fixed cost of each bin is set to its capacity and the unit cost is randomly picked from the interval $[0,1]$.
For each combination of  $(n,m) \in \{(15,10),$ $(25,15),$ $(25,25),$ $(200,10),$ $(250,15),$ $(500,30) \}$ 
and $X \in \{1,2,3\}$  we generated $10$ instances 
giving 180 instances in total. 

\begin{table}[t]
\centering%
\caption{Comparison of the results obtained using  MIP. \arcflow. and CP approaches approaches on random bin packing with
usage cost problem instances with 600 seconds time-limit.\label{tab:cbp}}
\scalebox{0.67}{
\begin{tabular}{|r|r|r|r|r|r|r|r|r|r|r|r|r|r|r|r|c|r|c|c|c|}             
\hline
n & m & X &\multicolumn{4}{|c|}{MIP}    	& \multicolumn{3}{c|}{CP}  & \multicolumn{3}{c|}{CP+CG}  &  \multicolumn{4}{c|}{\arcflow\ }     & CG\\ \hline
n	&	m	&	X	& gap $z^{*}_1$	&	\#ns	&	cpu	&	nodes	&	\#ns	&	cpu	&	nodes	&	\#ns	&	cpu	&	nodes	& gap $z^{*}_3$	&	\#ns	&	cpu	&	nodes	&gap $z^{*}_2$	\\
\hline
15	&	10	&	1	&	4.81	&	10	&	1.2		&	2112		&	10	&	0.3		&	1270		&	10	&	0.3		&	\textbf{1109}		&	4.53	&	10	&	2.1	&	2402		&	\textbf{4.47}	\\ 
15	&	10	&	2	&	2.91	&	10	&	1.1		&	2050		&	10	&	0.1		&	1020		&	10	&	0.1		&	\textbf{785}		&	1.81	&	10	&	0.7	&	1508		&	\textbf{1.77}	\\ 
15	&	10	&	3	&	3.31	&	10	&	0.8		&	1847		&	10	&	0.4		&	4372		&	10	& 	\textbf{0.2}		&	972		&	1.36	&	10	&	0.6	&	\textbf{420}		&	\textbf{1.33}	\\ \hline
25	&	15	&	1	&	1.76	&	10	&	35.1		&	29546	&	10	&	20.1		&	140533	&	10	&	\textbf{1.2}		&	1410		&	1.60	&	9	&	47.5	&	\textbf{891}		&	\textbf{1.59}	\\ 
25	&	15	&	2	&	1.91	&	10	&	74.2		&	114068	&	10	&	10.1		&	78297	&	10	&	\textbf{0.8}		&	\textbf{1498}		&	1.53	&	10	&	61.2	&	12557	&	\textbf{1.51}	\\
25	&	15	&	3	&	2.60	&	10	&	22.6		&	46350	&	8	&	26.9		&	198792	&	10	&	\textbf{7.3}		&	3359		&	0.93	&	10	&	11.3	&	\textbf{1039}		&	\textbf{0.93}	\\ \hline
25	&	25	&	1	&	1.67	&	9	&	19.0		&	10876	&	8	&	14.5		&	99250	&	\textbf{10}	&	\textbf{1.3}		&	\textbf{1394}		&	1.42	&	9	&	31.2	&	4582		&	\textbf{1.42}	\\
25	&	25	&	2	&	1.47	&	9	&	97.4		&	114530	&	9	&	17.6		&	117327	&	10	&	\textbf{2.1}		&	\textbf{1706}		&	1.08	&	10	&	54.8	&	5967		&	\textbf{1.07}	\\
25	&	25	&	3	&	2.30	&	10	&	33.7		&	53462	&	8	&	109.5	&	537144	&	10	&	\textbf{2.2}		&	1816		&	0.93	&	10	&	20.6	&	\textbf{510}		&	\textbf{0.92}	\\\hline
200	&	10	&	1	&	2.73	&	10	&	\textbf{1.9}		&	\textbf{1011}		&	10	&	6.4		&	14626	&	10	&	13.5		&	14626	&	2.73	&	0	&	-	&	-		&	2.73	\\
200	&	10	&	2	&	2.16	&	10	&	18.9		&	68121	&	10	&	\textbf{2.9}		&	6641		&	10	&	10.3		&	6641		&	2.16	&	0	&	-	&	-		&	2.16	\\
200	&	10	&	3	&	1.82	&	10	&	16.2		&	93306	&	10	&	\textbf{0.4}		&	4064		&	10	&	1.7		&	4064		&	1.82	&	0	&	-	&	-		&	1.82	\\\hline
250	&	15	&	1	&	1.53	&	10	&	\textbf{3.9}		&	\textbf{1627}		&	10	&	7.4		&	11159	&	10	&	11.7		&	11159	&	1.53	&	0	&	-	&	-		&	1.53	\\
250	&	15	&	2	&	1.30	&	9	&	26.3		&	50275	&	10	&	\textbf{6.1}		&	12902	&	10	&	20.4		&	12902	&	7.21	&	0	&	-	&	-		&	1.30	\\
250	&	15	&	3	&	0.96	&	4	&	190.9	&	546599	&	10	&	\textbf{4.8}		&	14111	&	10	&	23.7		&	14111	&	0.96	&	0	&	-	&	-		&	0.96	\\\hline
500	&	30	&	1	&	0.62	&	10	&	\textbf{18.1}		&	7449		&	10	&	30.6		&	14829	&	10	&	77.9		&	14829	&	0.62	&	0	&	-	&	-		&	0.62	\\
500	&	30	&	2	&	0.46	&	7	&	173.2	&	193274	&	10	&	\textbf{16.7}		&	14340	&	10	&	78.4		&	14340	&	0.46	&	0	&	-	&	-		&	0.46	\\
500	&	30	&	3	&	0.28	&	1	&	184.6	&	322806	&	\textbf{10}	&	\textbf{24.2}		&	21285	&	9	&	187.9	&	19678	&	0.28	&	0	&	-	&	-		&	0.28	\\\hline
\end{tabular}
}
\end{table}

The time-limit  was 600 seconds. All the experiments were carried out  on a Dual Quad Core Xeon CPU, running
Linux 2.6.25 x64, with  11.76 GB of RAM, and 2.66 GHz processor speed. The LP solver used was CPLEX 12.5 (default parameters) and the CP solver was Choco 2.1.5. Table~\ref{tab:cbp} reports the average cpu time in seconds (denoted \texttt{cpu}), the average number of nodes in the search tree (denoted \texttt{nodes}) and the average gap of the lower-bounds found at root node (denoted $gap\:z^*_x$). This gap is computed as a percentage of the best known solution found. Column \texttt{\#ns} gives the number of instances solved to optimality (\emph{i.e} optimal value was found and proved optimal) within the time limit.


The CP approach shows better performance when scaling to larger size instances (and capacities) than the MIP and \arcflow\ models. Overall, the  \arcflow\ model fails to solve optimally 92 instances of the 180 instances, MIP fails on 21, CP fails on 7 and CP+CG only on 1.

On one side, the search space is dramatically reduced by the propagation of $z_2^*$ (CP+CG) for small problems. On the other side, the bound becomes ineffective on large problems (this can also be seen on the gap at the root node). \\

\section{Conclusion and Future Work}
\label{sec-conclusions}
 
Bin Packing with Usage Costs (BPUC) problem can be viewed as a core subproblem of many optimisation problems related with workload consolidation in data centres. 
 The main contribution of this paper is  the study of various lower bounds and exact formulation for BPUC. Firstly, the value of the linear relaxation of a basic LP model for BPUC can be easily computed. Secondly, we show that this bound can be strengthened with filtering algorithms reasoning on the minimum/maximum load of each bin and that the resulting CP approach can efficiently handle relatively large instances.

\bibliographystyle{elsarticle-harv}

\end{document}

%% file: introduction.tex
Data centres are a critical and ubiquitous resource for providing infrastructure for banking, Internet and electronic commerce. 
They use enormous amounts of electricity,
 and this demand is expected to increase in the future.
For example, a report by the \emph{EU Stand-by Initiative} stated that in 2007 Western European data centres consumed 56 Tera-Watt Hours (TWh) of power, which is expected to almost double to 104 TWh per year by 2020.\footnote{\scriptsize \url{http://re.jrc.ec.europa.eu/energyefficiency/html/standby_initiative_data_centers.htm}}
Energy consumption is one of the most important sources of expense in data centres 
and energy efficiency is at the core of their competitive advantage.
 The ongoing increase in energy prices  (a 50\% increase is forecasted by the French senate by 2020) 
and the growing market for cloud computing are the  main incentives for the design of energy efficient centres.
Minimising energy consumption is not only important for economic reason but also for environmental reason. 
The consulting firm McKinsey reported that only 6-12\% of electricity used
 by data centres can be attributed to the performance of productive
 computation~\cite{nytimes}.
Hence, there is a lot of scope for reducing the carbon footprint of data centres
by utilising  servers more efficiently.

 We study a problem associated with the EnergeTIC\footnote{Minalogic EnergeTIC is a Global competitive cluster located in Grenoble
France and fostering research-led innovation in intelligent
miniaturised products and solutions for industry.  
} project  which was accredited by the French government (FUI)~\cite{energetic} that brought together four companies (Bull, Business \& Decision Eolas, Schneider Electric, UXP),  several academic partners (G2Elab, G-SCOP, LIG),  and a startup company (Vesta-System). 
The objective  is to control the energy consumption of a data
center and ensure that it is consistent with application needs, economic
constraints and service level agreements.
We focus on  how to reduce  energy cost by taking   {\scshape cpu} requirements of client applications, IT equipment and  virtualisation techniques into account. 


The EnergeTIC  optimisation problem for  energy-cost aware data centre
assignment systems is about  allocating virtual machines   demands
 to servers where  the energy cost per unit of computation can vary  between
  different servers.

The problem can be defined by 
 a set of servers and a set of virtual applications  to be run on those servers.
Each server is associated with a set of available resources, e.g. CPU, RAM, DISK etc.
Each virtual application  is associated with a set of required resource values.
The solution of the problem 
is an assignment of virtual machines to servers  which respects a set of hard constraints. 
The objective is to 
 take advantage of  differences in energy costs across the servers, 
 the requirements of virtual applications, 
and consolidate machine workload  to ensure that servers are well utilised so that energy costs can be reduced.

A combinatorial optimisation model for the problem of loading servers to a desired utilisation level has, at its core,  a variant of bin packing (BP) problem~\cite{Srikantaiah:2008}.
In such a model each server is represented by a bin with a capacity equal to the amount of resource available. 
Bin packing is a very well studied NP-Hard problem. 
In the present context, energy consumption cannot be accurately modelled by only considering  the number of active servers
 since the energy cost of a server is  
 a function of the workload.
Furthermore, servers require energy not only to run processes but also when they are idle which must be considered in the energy management process.
We  focus on an extension of bin packing problem with linear costs associated with the use of bins in Section~\ref{bplincost}. 
This is a key sub-problem of the application domain and we show how to handle it efficiently with Constraint Programming (CP). 
We study lower bounds based on Linear Programming 
and extend the bin packing global constraint with  cost information. The present work has been partially published in \cite{DBLP:conf/cp/CambazardMOS13} which focus on the real-life application.

%% file: bpuc.tex
In this section we consider a  variant of the Bin Packing problem (BP)  \cite{Srikantaiah:2008},  which is the key sub-problem of the application  investigated here. 
We denote by $S = \{w_1, \ldots, w_n\}$ the integer sizes of the $n$ items such that $w_1 \leq w_2 \leq \ldots w_n$. A bin $j$ is characterised by an integer capacity $C_j$, a non-negative fixed cost $f_j$ and a  non-negative  cost $c_j$ for each unit of used capacity. We denote by $B=\{\{C_1, f_1, c_1\},\ldots, \{C_m, f_m, c_m\}\}$ the characteristics of the $m$ bins. A bin is used when it contains at least one item. Its cost is a linear function $f_j + c_jl_j$, where $l_j$ is  the total size of the items  in bin $j$. The total load  is denoted by $W = \sum^{n}_{i = 1} w_i$ and the maximum capacity by $C_{max} = max_{1 \leq j \leq m} C_j$. The problem is to assign each item to a bin subject to the capacity constraints so that we minimise  the sum of the costs of all bins. We refer to this problem as the \emph{Bin Packing with Usage Cost} problem (BPUC).
BP  is a special case of BPUC where all $f_j$ are set to 1 and all $c_j$ to 0. The following example shows that a good solution for BP might not yield a good solution for  BPUC. 

\begin{example}\label{ex1} In Figure~\ref{example1}, Scenario~1,   B =\{(9,0,1),(3,0,2),(3,0,2),(3,0,2),(3,0,2)\} and S = \{2,2,2,2,3,3,3\}. Notice that $ \forall j, f_j = 0$. The packing $(P_1)$ : \{\{2,2,2,2\}, \{3\}, \{3\}, \{3\}, \{\}\} is using the minimum number of bins and has a  cost of 26 (8*1 + 3*2 + 3*2 + 3*2). The packing $(P_2)$: \{\{3,3,3\}, \{2\}, \{2\}, \{2\}, \{2\}\} is using one more bin but has a  cost of 25 (9 + 2*2 + 2*2 + 2*2 + 2*2). Here, ($P_2$) is better than ($P_1$) and using the minimum number of bins is not a good strategy. 
Now  change the last unit cost to $c_5=3$ (see Figure~\ref{example1}, Scenario 2). The cost of $(P_1)$ remains unchanged since it does not use bin number 5 but the cost of $(P_2)$ increases to 27, and thus  $(P_1)$ is now better than $(P_2)$.
\end{example}
\begin{figure}[h]
\centering
\includegraphics[width=12cm]{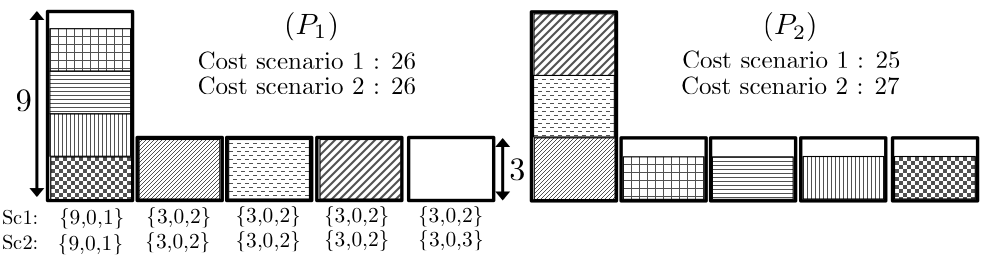}
\caption{Example of optimal solutions in two scenarios of costs. In Scenario 1, the best solution has no waste on the cheapest bin. In Scenario 2, it does not fill completely the cheapest bin.}
\label{example1}
\end{figure}

\noindent\textbf{Literature Review.} 
A first relevant extension of BP for the current paper is called Variable Size Bin-Packing, where bins have different capacities and the problem is to minimise the sum of the wasted space over all used bins~\cite{Monaci2002}. It can be seen as a special case of BPUC where all $f_j = C_j$ and $c_j = 0$. 
Recent lower bounds and an exact approach are examined in~\cite{Haouari:2011}.
A generalisation to any kind of fixed cost is presented in~\cite{Crainic2011}, which can be seen as a special case of BPUC where all $c_j = 0$. 
Concave costs of bin utilisation  studied in~\cite{Li2006} are more general than the linear cost functions of BPUC. 
However~\cite{Li2006} does not consider bins of different capacities and deals with the performance of classical BP heuristics whereas we are focusing on lower bounds and exact algorithms.
Secondly, BP with general cost structures have been introduced in~\cite{AnilyBS94} and studied in~\cite{Epstein:2012}. The authors investigated BP with non-decreasing and concave cost functions of the number of items put in a bin. They extend it with profitable optional items in~\cite{Baldi2012}. Their framework can model a fixed cost but does not relate to bin usage.


%

\newenvironment{equ}{\vspace{0.1cm}\begin{equation}}{\end{equation}\vspace{1mm}}

\section{Basic Formulation and  Lower Bounds of BPUC}
 Numerous linear programming models have been proposed for BP ~\cite{DBLP:journals/eor/Carvalho02}. We first present a 
 formulation for  BPUC.
  For each bin a binary variable $y_j$ is set to 1 if bin $j$ is used in the packing,
and 0 otherwise. For each item $i \in \{1,\ldots,n\}$ and each bin $j \in \{1,\ldots,m\}$  a binary variable $x_{ij}$ is set to 1 if item $i$ is packed into bin $j$, and 0 otherwise. 
We add non-negative variables $l_j$ representing the load  of each bin j. The  model is as follows:
\begin{equation}
\label{lp-model}
     \begin{array}{ll}
     \textbf{Minimize} \qquad z_1 = \sum_{j=1}^{m} (f_jy_j + c_jl_j) & \\     
	(1.1) \qquad \sum_{j=1}^m x_{ij}=1, & \forall i \in \{1,\ldots,n\} \\
	(1.2) \qquad \sum_{i=1}^{n} w_ix_{ij} = l_j, &  \forall j \in \{1,\ldots,m\} \\
	(1.3) \qquad l_j \leq C_jy_j , & \forall j \in \{1,\ldots,m\}  \\
	(1.4) \qquad x_{ij} \in \{0,1\},  y_j \in \{0,1\},  l_j \geq 0 \:\:\:&  \forall j \in \{1,\ldots,m\}, \forall i \in \{1,\ldots,n\} \\
     \end{array}
\end{equation}
Constraint (1.1) states that each item is assigned to  one bin whereas (1.2) and (1.3) enforce the capacity of the bins. We now  characterise the linear relaxation of the model. 
Let $r_j = f_j/C_j  + c_j$ be a real number associated with bin $j$. If bin $j$ is filled completely, $r_j$ is the cost of one unit of space in bin $j$. We sort the bins by non-decreasing $r_j$: $r_{a_1} \leq r_{a_2} \leq \ldots \leq r_{a_m}$; $a_1, \ldots, a_m$ is a permutation of the bin indices $1,\ldots, m$. Let  $k$ be the minimum number of bins such that $ \sum^{k}_{j = 1} C_{a_j} \geq W$.
\begin{proposition}
Let $z_1^*$ be the optimal value of the linear relaxation of the formulation~(1). We have $z_1^{*} \geq Lb_1$ with 
$Lb_1 = \sum^{k-1}_{j = 1} C_{a_{j}}r_{a_{j}} + (W - \sum^{k-1}_{j = 1} C_{a_j})r_{a_{k}}$. 
\end{proposition}
 \begin{proof}
 $z_1^* = \sum_{j=1}^{m} (f_jy_j + c_jl_j)  \geq \sum_{j=1}^{m} (f_j\frac{l_j}{C_j} + c_jl_j)$ because of the constraint $l_j \leq C_jy_j$, so $z_1^* \geq \sum_{j=1}^{m} (\frac{f_j}{C_j} + c_j)l_j \geq \sum_{j=1}^{m} r_jl_j $. $Lb_1$ is the quantity minimising $\sum_{j=1}^{m}r_jl_j$ under the constraints $\sum_{j} l_j = W$ where each $l_j \leq C_j$. To minimise the quantity  we must split W over the $l_j$ related to the smallest coefficients $r_j$. Hence, $z_1^* \geq \sum_{j=1}^{m} r_jl_j \geq Lb_1$. \qed
 \end{proof} 
$Lb_1$ is a lower bound of BPUC  that can be computed in $O(mlog(m) + n)$. Also notice that $Lb_1$  is  the bound that we get by solving the linear relaxation of formulation (1).
\begin{proposition}
 $Lb_1$ is the optimal value of the linear relaxation of the formulation~(1). 
\end{proposition}
 \begin{proof} For all $j < k$, we set each $y_{a_j}$ to $1$ and $l_{a_j}$ to $C_j$.  We fix $l_{a_k}$ to $(W - \sum^{k-1}_{j = 1} C_{a_j})$ and $y_{a_k}$ to $l_{a_k}/C_{a_k}$. For all $j > k$ we set $y_{a_j} = 0$ and $l_{a_j} = 0$. Constraints (1.3) are thus satisfied. Finally we fix $x_{i,a_j} = \frac{l_{a_j}}{W}$ for all $i,j$ so that constraints (1.2) and (1.1) are satisfied. This is a feasible solution of the linear relaxation of (1) achieving an objective value of $Lb_1$. We have, therefore, $Lb_1 \geq z_1^*$ and consequently $z_1^*=Lb_1$ from Proposition~1.  \qed
 \end{proof}
Adding the constraint $x_{ij} \leq y_j$ for each item $i$ and bin $j$, strengthens the linear relaxation only if $W <  C_{a_k}$. Indeed, the solution given in the proof is otherwise feasible for the constraint, ($ \forall j < k,  \: x_{i,a_j} = \frac{l_{a_j}}{W}\leq y_{a_j} = 1$ and for $j=k$ we have $\frac{l_{a_k}}{W}  \leq \frac{l_{a_k}}{C_{a_k}}$ if $W \geq C_{a_k}$). 


\paragraph{Dominance and symmetries.} Formulation~(1) can be used to solve the problem exactly and strengthened to take some bin and item symmetries into account. Typically, if bins $i$ and $j$ are such that $f_i \leq f_j$, $c_i \leq c_j$ and $C_i \geq C_j$ then bin $i$ is said to dominate bin $j$ and the constraints $y_j \leq y_i$ as well as $l_j \leq l_i$ can be added to rule out some dominated solutions. 
Suppose now that we have $k$ items that have the same size, any permutation of the items in a feasible solution remains feasible and have the same overall cost. A simple way to break this symmetry is to replace for each bin $j$, the k variables $x_{ij} \in \{0,1\}$  corresponding to the identical items by a variable $x^{'}_{j} \in \{0,\ldots,k\}$ giving the number of the k items assigned to $j$ (and adding $\sum_{j} x^{'}_{j} = k$). We also perform in practice another simple pre-processing by tightening the capacities of the bins using dynamic programming. Given the item sizes and for each bin $i$, one can compute efficiently (when $C_i$ is sufficiently small) the largest integer less than or equal to $C_i$ that is also equal to the sum of a subset of $S$.

\section{Two Extended Formulations of BPUC}
\label{af}

\paragraph{\textbf{The Cutting Stock Model}}
The  formulation of Gilmore and Gomory for the cutting stock problem~\cite{gilmore1961linear} can  be  adapted for BPUC. The items of equal size are now grouped and for $n^{'} \leq n$ different sizes we denote the number of items of sizes $w^{'}_1, \ldots, w^{'}_{n^{'}}$  by $q_1, \ldots, q_{n^{'}}$  respectively.
A cutting pattern for bin $j$ is a combination of item sizes that fits into bin $j$ using no more than $q_d$ items of size $w^{'}_{d}$. 
In the $i$-th pattern of bin $j$, the number of items of size $w^{'}_d$ that are in the pattern is denoted $g_{dij}$. 
Let $I_j$ be the set of all patterns for bin $j$. 
The cost of the $i$-th pattern (assumed to be non empty) of bin $j$ is therefore equal to $co_{ij} = f_{j} + (\sum^{n^{'}}_{d=1} g_{dij}w^{'}_{d})c_j$. The cutting stock formulation is using a variable $p_{ij}$ for the $i$-th pattern of bin $j$:
%
\begin{equation}
\label{bp-model}
     \begin{array}{ll}
     \textbf{Minimize} \qquad z_2 = \sum_{j=1}^{m}\sum_{i \in I_j} co_{ij}p_{ij} & \\          
	(2.1) \qquad  \sum^{m}_{j = 1}\sum_{i \in I_j} g_{dij}p_{ij} = q_d & \forall d \in \{1,\ldots,n^{'}\} \\
	(2.2) \qquad  \sum_{i \in I_j} p_{ij} = 1 & \forall j \in \{1,\ldots,m\} \\
	(2.3) \qquad p_{ij} \in \{0,1\} & \forall j \in \{1,\ldots,m\}, i \in I_j  \\
\end{array}
\end{equation}
Constraint (2.1) states that each item has to appear in a pattern (thus in a bin) and (2.2) enforces one pattern to be designed for each bin (convexity constraints). A pattern $p_{ij}$ for bin $j$ is valid if $\sum^{n^{'}}_{d = 1} g_{dij}w^{'}_{d} \leq C_j$ and all $g_{dij}$ are integers such that $q_d \geq g_{dij} \geq 0$. The sets $I_j$ have an exponential size so the linear relaxation of this model can be solved using column generation. The pricing step is a knapsack problem that can be solved efficiently by dynamic programming if the capacities are small enough. The pricing problem for bin $j$ can be written as follows (where $\pi_d$ is the dual variable associated to the d-th constraint (2.1) and $\lambda_j$ is the dual variable related to j-th constraint (2.2)):
\begin{equation*} 
     \begin{array}{ll}
     \textbf{Minimize} \qquad f_{j}y + (\sum^{n^{'}}_{d=1} g_{d}w^{'}_{d})c_j - \sum^{n^{'}}_{d = 1} \pi_{d}g_{d} - \lambda_j& \\          
	 \qquad  \sum^{n^{'}}_{d = 1} g_{d}w^{'}_{d} \leq C_jy &\\
	 \qquad g_{d} \in \{0, 1, \ldots, q_d\} & \forall d \in \{1,\ldots,n^{'}\} \\
	 \qquad y \in \{0, 1\} & 
	\end{array}
\end{equation*}
In practice, we solve the pricing problem using Dynamic Programming in $O(nC_j)$.
\paragraph{\textbf{The \arcflow\ Model}}
Carvalho introduced an elegant \arcflow\ model
 for BP~\cite{Carvalho99,DBLP:journals/eor/Carvalho02}. His model explicitly uses each unit of capacity  of the bins. In the following we show how to adapt it for BPUC.
Consider a multi-graph $G(V,A)$, where $V = \{0,1,...,C_{max}\} \cup \{F\}$ is the set  of $C_{max}+2$ nodes labelled from $0$ to $C_{max}$ and a final node labelled $F$, and  $A=I \cup J$ is the set of  two kinds of edges. An edge  $(a,b) \in I$ between two nodes labelled $a \leq C_{max}$ and $b \leq C_{max}$ represents the use of an item of size $b-a$. 
 An edge of $(a,F) \in J$ for each bin $j$ represents the usage $a$ of the bin $j$, and therefore $a \leq C_j$. An example of such a graph is shown in Figure~\ref{example2}(a).
%
%
\begin{figure}[t!]
\centering
\includegraphics[width=12cm]{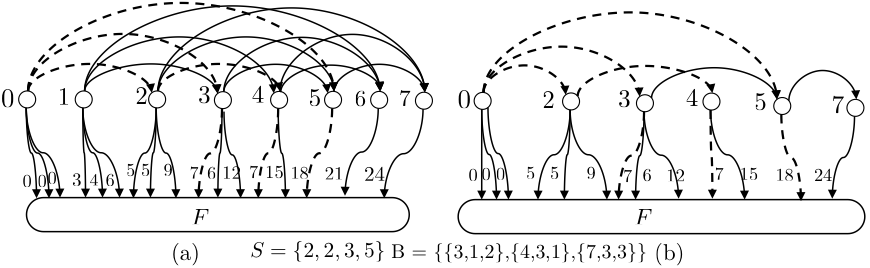}
\caption{(a) An example of the graph underlying the \arcflow\ model for $S= \{2,2,3,5\}$, $B=\{\{3,1,2\}, \{4,3,1\}, \{7,3,3\}\}$ so that $C_{max} = 7$. A packing is shown using a dotted line: $\{3\}$ is put in the first bin for a cost of 7, $\{2,2\}$ is in the second bin for a cost of 7 and $\{5\}$ in the last bin for a cost of 18.
(b) The graph underlying the \arcflow\ model after the elimination of symmetries.}
\label{example2}
\end{figure}
Notice that this formulation has symmetries since a packing can be encoded by many different paths.
Some reduction rules were given by Carvalho~\cite{Carvalho99}, which help in reducing such symmetries (see Figure~\ref{example2}(b)).



BPUC can be seen  as a minimum cost flow between 0 and $F$ with constraints enforcing the number of edges of a given length used by the flow to be equal to the number of items of the corresponding size. We have variables $x_{ab}$ for each edge $(a,b) \in I$ as well as variables $y_{aj}$ for each pair of bin $j\in \{1,\ldots,m\}$ and $a \in V$. 
The  cost of using an edge $(a,F) \in J$ for bin $j$ with $a > 0$ is $co_{aj} = f_{j} + a\cdot c_j$ and $co_{0j} = 0$. The model is as follows:
\begin{equation}
\begin{footnotesize}
     \begin{array}{l}
     \textbf{Minimize} \qquad z_3 = \sum^{m}_{j=1} \sum^{k = C_{max}}_{k = 0} co_{kj} y_{kj}  \\     
(3.1) \:\:\: \sum_{(a,b) \in A} x_{ab} - \sum_{(b,c) \in A} x_{bc} -  \sum^{m}_{j = 1} y_{bj} =  \left\{
    \begin{array}{ll}
        0 & \forall b \in \{1,2, \ldots, C_{max}\} \\
        -m & \mbox{for } b = 0\\
    \end{array}
\right.   \\
(3.2) \:\:\: \sum^{C_j}_{a=0} y_{aj} = 1 \qquad  \qquad\qquad\qquad\forall j \in \{1, \ldots, m\}  \\
(3.3) \:\:\: \sum_{(k,k+w^{'}_{d}) \in A } x_{k,k+w^{'}_{d}} = q_{d} \qquad \qquad \forall d \in \{1,2, \ldots, n^{'}\} \\
(3.4) \:\:\: y_{aj} = 0 \qquad \qquad\qquad \qquad  \qquad\: \forall (j,a)\in \{1, \ldots, m\} \times \{C_{j}+1,\ldots,C_{max}\} \\
(3.5) \:\:\: x_{ab} \in \mathbb{N}\qquad \qquad  \qquad\qquad\qquad \forall (a,b) \in A \\
(3.6) \:\:\: y_{aj} \in \{0,1\}  \qquad\qquad\qquad  \qquad\:\:\:\: \forall (j,a)\in \{1, \ldots, m\} \times \{0,\ldots,C_{max}\}
\end{array}
\end{footnotesize}
\end{equation}
Constraint (3.1) enforces the flow conservation at each node, and Constraint (3.2) states that each bin should be used exactly once. Constraint (3.3) ensures that all the items are packed, while Constraint (3.4) enforces that bin $j$ is not used beyond its capacity $C_j$. A solution can be obtained again by decomposing the flow into paths. The number of variables in this model is in $\mathcal{O}((n^{'}+ m) \cdot C_{max})$ and the number of constraints is $\mathcal{O}(C_{max} + m + n^{'})$. Although its
LP  relaxation is stronger than that of Model~(1), it remains dominated by that of Model~(2).
\begin{proposition}$z^*_3 \leq z^*_2$. The optimal value of the linear relaxation of (3) is less than  the optimal value of the linear relaxation of (2). 
\end{proposition}
 \begin{proof} Let $(p^*)$ be a solution of the linear relaxation of (2). Each pattern $p^*_{ij}$ is mapped to a path of the \arcflow\ model. A fractional value $p^*_{ij}$ is added on the arcs corresponding to the item sizes of the pattern (the value of the empty patterns for which all $g_{dij} = 0$ is put on the arcs $y_{0j}$). 
The flow conservation (3.1) is satisfied by construction, so is (3.2) because of (2.2) and so are the demand constraints (3.3) because of (2.1). Any solution of (2) is thus encoded as a solution of (3) for the same cost so $z^*_3 \leq z^*_2$. \qed
 \end{proof}
%

\begin{proposition} $z^*_2$ can be stronger than $z^*_3$ \emph{i.e.} there exist instances such that $z^*_2 > z^*_3$.  
\end{proposition}

 \begin{proof} Consider the following instance: $S=\{1,1,2\}$ and $B=\{\{3,1,1\}, \{3,4,4\}\}$. Two items of size 1 occurs so that $n^{'} = 2$, $q_1 = 2, q_2 = 1$ corresponding to $w^{'}_1=1, w^{'}_2=2$. The two  bins have to be used and the first \emph{dominates} the second (the maximum possible space is used in bin 1 in any optimal solution) so the optimal solution is the packing $\{\{2,1\},\{1\}\}$ (cost of 12). Let's compute the value of $z^{*}_2$. It must fill the first bin  with the pattern $[g_{111},g_{211}] = [1,1]$ for a cost of $4$. Only three possible patterns can be used to fill the second bin: $[0,0]$, $[1,0]$ and $[2,0]$ (a valid pattern $p_{i2}$ is such that $g_{1i2} \leq 2$). The best solution is using $[g_{112},g_{212}]=[2,0]$ and $[g_{122},g_{222}]=[0,0]$ taking both a 0.5 value to get a total cost $z^{*}_2 = 4 + 6 = 10$. The \arcflow\ model uses a path to encode the same first pattern [1,1] for bin 1. But it can build a path for bin 2 with a $\frac{1}{3}$ unit of flow taking three consecutive arcs of size 1 to reach a better cost of $\frac{1}{3}*16 \approx 5,33$. This path would be a pattern [3,0] which is not valid for (\ref{bp-model}). So $z^*_3 \approx 9.33$ and $z^*_2 > z^*_3$. 
 \qed
 \end{proof}
The \arcflow\ model  may use a path containing more than $q_d$ arcs of size $w^{'}_d$ with a positive flow whereas no such patterns exist in (3) \textit{because} the sub-problem is subject to the constraint $ 0 \leq g_{dij} \leq q_{d}$. The cutting stock formulation used in \cite{Carvalho99}  ignores this constraint and therefore the bounds are claimed to be equivalent.

\section{Extending the Bin Packing Global Constraint}

A bin packing global constraint was introduced for constraint programming  by \cite{DBLP:conf/cp/Shaw04} and discussions about its filtering can be found in \cite{Schaus2009,DBLP:conf/aaai/ReginR11}. We present an  extension of this global constraint to handle BPUC. The scope and parameters are as follows:
$$\textsc{BinPackingUsageCost}([x_{1},\ldots,x_{n}], [l_{1},\ldots,l_{m}], [y_{1}, \ldots, y_{m}], b, z, S, B)$$
Variables $x_i \in \{1,\ldots,m\}$, $l_j \in [0,\ldots,C_{j}]$ and $b \in \{1,\ldots,m\}$ denote  the bin assigned to item $i$, the load of bin $j$, and the number of bins used, respectively. These  variables are also used by the \textsc{BinPacking} constraint. Variables $y_i \in \{0,1\}$ and $z \in \mathds{R}$ are due to the cost. They denote whether bin $j$ is open, and the cost of the packing. The last two arguments refers to BPUC and give the size of the items as well as the costs (fixed and unit). In the following, $\underline{x}$  (resp. $\overline{x}$) denotes the lower (resp. upper) bound of variable $x$.   \\

\noindent\textbf{Cost-based Propagation using $Lb_1$.}
The characteristics of the bins of the restricted  BPUC  problem  based on the current state of the domains of the variables is  denoted by $B^{'}$, and defined by $B^{'} = \{\{C^{'}_1, f^{'}_1, c_1\},\ldots, \{C^{'}_m, f^{'}_m, c_m\}\}$ where $C^{'}_j = \overline{l_j} - \underline{l_j}$ is the remaining capacity, and $f^{'}_j$ is the remaining fixed cost $f^{'}_j= (1-\underline{y_j})f_j$ that is set to 0 if the bin is known to be open. The total load that remains to be allocated to the bins is denoted $W^{'} = W - \sum^{m}_{j=1}\underline{l_j}$. Notice that we use the lower bounds of the loads rather than the already packed items. We assume it is strictly better due to the reasoning of the bin packing constraint.

\paragraph{Lower bound of z} 
The first propagation rule is the update of the lower bound   of $z$ denoted by $\underline{z}$. The bound is computed by summing the cost due to open bins and minimum loads with the value of $Lb_1$ on the remaining problem. 
It gives a maximum possible increase in cost denoted by $gap$:
\begin{equation}
\label{r1}
\begin{array}{ccc}
Lb^{'}_1 = \sum^{m}_{j=1}(\underline{l_j}c_j + \underline{y_j}f_j) + Lb_1(W^{'},B^{'}); & \hspace{6mm} \underline{z} \leftarrow max(\:\underline{z}, Lb^{'}_1); & \hspace{6mm} gap = \overline{z} - Lb^{'}_1
\end{array}
\end{equation}

\noindent\emph{Bounds of the load variables.} 
We use the notation $L_j$ to denote the units of space used by $Lb_1$  on bin $a_j$. The bins $a_1, \ldots, a_{k-1}$ are fully used so $\forall j < k, L_j = C^{'}_{a_j}$, for bin $a_k$ we have $L_k = W^{'} - \sum^{k-1}_{j = 1} C^{'}_{a_j}$ and $\forall j > k, L_j = 0$.
We define the bin packing problem $B^{''}$ obtained by excluding the space supporting the lower bound $Lb_1(W^{'},B^{'})$. 
 The resulting bins are defined as $B^{''} = \{\{C^{''}_1, f^{'}_1, c_1\},\ldots, \{C^{''}_m, f^{'}_m, c_m\}\}$ where $C^{''}_{a_j} = 0$ for all $j < k$, $C^{''}_{a_k} = C^{'}_{a_k} - L_k$ and $C^{''}_{a_j} = C^{'}_{a_j}$ for all $j > k$. Lower and upper bounds  of loads are adjusted with rules ($\ref{r2}$). 


Let $q_{a_j}^-$ be the largest quantity that can be removed from a bin $a_j$, with $j \leq k$,
and put at the other cheapest possible place without overloading $\overline{z}$. Consequently, when $j < k$, $q_{a_j}^-$ is the largest value in $[0,L_j]$ such that $(Lb_1(q_{a_j}^-,B^{''}) - q_{a_j}^-r_{a_j}) \leq gap$. When $j = k$, the same reasoning can be done by setting $C^{''}_{a_k} = 0$ in $B^{''}$. 

Similarly, let $q_{a_j}^+$ be the largest value in $[0,C^{'}_{a_j}]$ that can  be put on a bin $a_j$, with $j \geq k$, without triggering a contradiction with the remaining gap of cost. $q_{a_j}^+$ is thus the largest value in $[0,C^{'}_{a_j}]$ such that $(q_{a_j}^+r_{a_j} - (Lb_1(W^{'},B^{'}) - Lb_1(W^{'}-q_{a_j}^+,B^{'}))) \leq gap$.
\begin{equation}
\label{r2}
\begin{array}{cc}
\forall j \leq k, \:\:\: \underline{l_{a_j}} \leftarrow \underline{l_{a_j}} + L_j -q_{a_j}^-; & \hspace{35mm} \forall j \geq k, \:\:\: \overline{l_{a_j}} \leftarrow \underline{l_{a_j}} + q_{a_j}^+. \\
\end{array}
\end{equation}
\noindent\emph{Channelling.} The constraint ensures two simple rules relating the load and open-close variables (a bin of zero load can be open):
$y_j = 0 \implies l_j = 0 \: \textrm{ and } \: l_j>0 \implies y_j = 1.$

\smallskip\noindent\emph{Bounds of the open-close variables.} The propagation rule  for $\underline{l_j}$ can derive $\underline{l_j} > 0$ from (\ref{r2}), which in turn (because of the channelling between $y$ and $l$) will force a bin to open \emph{i.e.}  $y_{a_j} \in \{0,1\}$  will change  to $y_{a_j}=1$. To derive that a $y_j$ has to be fixed to $0$, we can use $Lb_1$ similarly to the reasoning presented for the load variables (checking that the increase of cost for opening a bin remains within the gap). 

Tightening the bounds of the load variables can trigger the existing filtering rules of the bin packing global constraint thus forbidding or committing items to bins. Notice that items are only increasing the cost indirectly by increasing the loads of the bins because the cost model is defined by the state of the bins (rather than the items). The cost-based propagation on $x$ is thus performed by the bin packing global constraint solely as a consequence of the updates on the bin related variables, i.e.\ $l$ and $y$.\\

\smallskip\noindent\emph{Filtering the load variables with dynamic programming.} Another filtering rule based on dynamic programming can be added to the Bin-Packing global constraint. We can simply update $\underline{l_j}$ (resp. $\overline{l_j}$) to the smallest (resp. largest) integer greater than (resp. less than) or equal to the current value of $\underline{l_j}$ (resp. $\overline{l_j}$) that can be reached with the remaining items that can still go on bin $j$. We can solve this problem by dynamic programming (taking into account the items already assigned to $j$). This is generally  very costly in practice and is not performed by default by the bin-packing global constraint. Nevertheless, in this case, one can notice that $Lb_1$ strongly relies on the accuracy of  $\underline{l_j}$  and $\overline{l_j}$ and we observed that it can help significantly in practice for some hard instances. This technique has been originally proposed by \cite{DBLP:journals/anor/Trick03} for knapsack constraints. \\

\begin{example}\label{ex2}Let's consider the following instance:  B =\{(9,9,5), (3,1,5), (7,14,3), (5,1,10), (12,12,10)\} and S = \{3,5,5,5\}. 
The total load is therefore $W = 18$ and the values or $r_j$ sorted increasingly are $r_3 = 14/7 + 3 = 5 < r_2 = 5.33 < r_1 = 6 < r_4 = 10.2 < r_5 = 11$. The lower bound is therefore $Lb_1 = 7\times5 + 3\times5.33 + 8\times6 = 99$. Figure~\ref{example2} is  visualising  $Lb_1$ and shows the optimal solution. Propagation on the bounds of the load variables by dynamic programming is not performed here for the illustrative purposes.

\begin{figure}[t!]
\centering
\includegraphics[width=11cm]{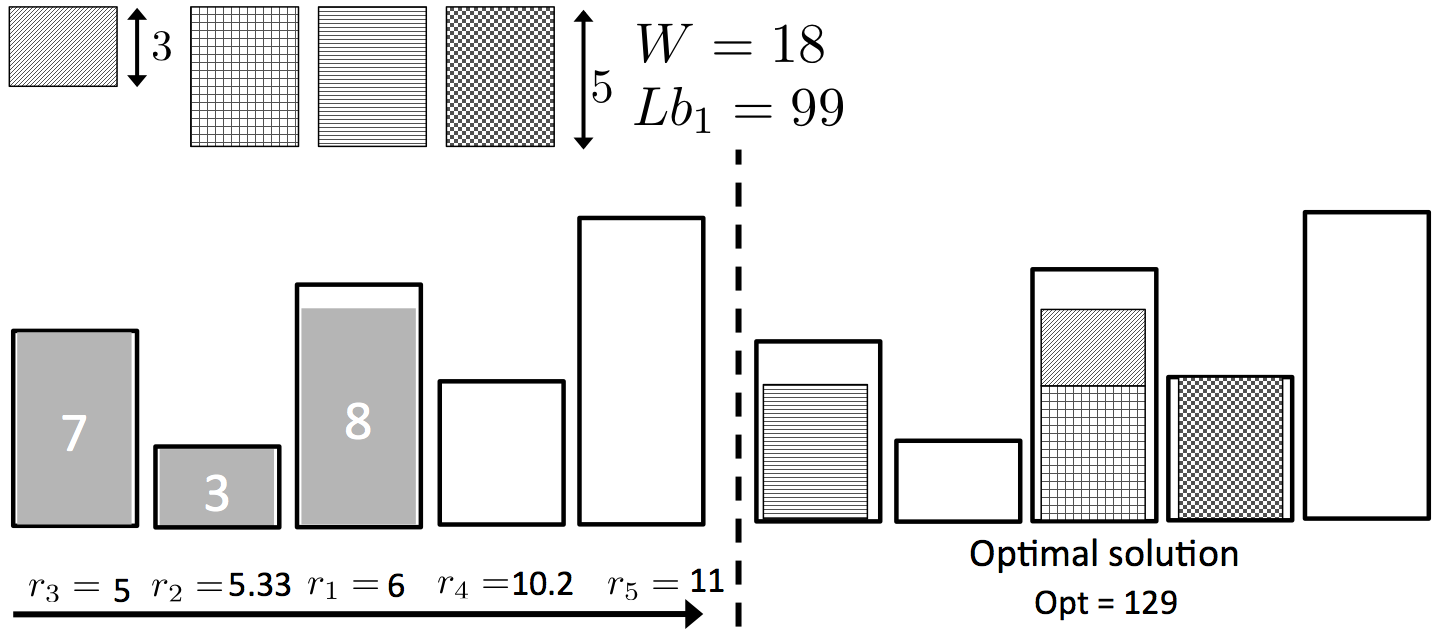}
\caption{Visualisation of the lower bound Lb$_1$ and optimal solution for the instance of example \ref{ex2}. The bins have been sorted by increasing $r$.} 
\label{example2}
\end{figure}

Assuming that we have an upper bound $\overline{z}$ of value 130, the gap is therefore initially equal to 130 - 99 = 31. 
Initially $Lb_1' = Lb_1$ as $W' = W$ and $B'=B$. 
Afterwards, the propagation is able to deduce $\underline{l_3} = \underline{l_1} = 1$ as well as $\overline{l_5} = 6$ as illustrated on the left of Figure~\ref{example3}. Indeed $l_3 = 0$ would lead to a lower bound of $3\times5.33 + 9\times6 + 5\times10.2 + 1*11 = 132$ thus overloading $\overline{z}$. Similarly $l_1 = 0$ would give $Lb_1=135 > 130$ and with $l_5 = 7$ we get $Lb_1= 134 > 130$. At this stage we know that bins 3 and 1 are open in any solution of cost less than 130. This is affecting the propagation as the fixed cost of bins 1 and 3 are now  included in $Lb^{'}_1$ and propagation is strengthened.
The  cost due to open bins is $9+14=23$ and due to  lower bounds of loads of bins 1 and 3 is $5+3=8$.
Consequently, $B'= \{(8,0,5), (3,1,5), (6,0,3), (5,1,10), (6,12,10)\}$ and $W'=16$.
Therefore, 
the $r$ values are now ranked: 
$r_3 = 3  <  r_1 = 5 < r_2 = 5.33 < r_4 = 10.2 < r_5 = 12$. Notice that the ordering on bins have changed.
The new value of $Lb_1$ is $6\times3 + 8\times5 + 2\times5.33 = 68.66$ and thus 
the new value of  $Lb'_1$ is $99.66$. 
When the fix point is reached we know that $\underline{l_3} = \underline{l_1} = 3$ as well as $\overline{l_5} = 3$.
 
\begin{figure}[t!]
\centering
\includegraphics[width=12cm]{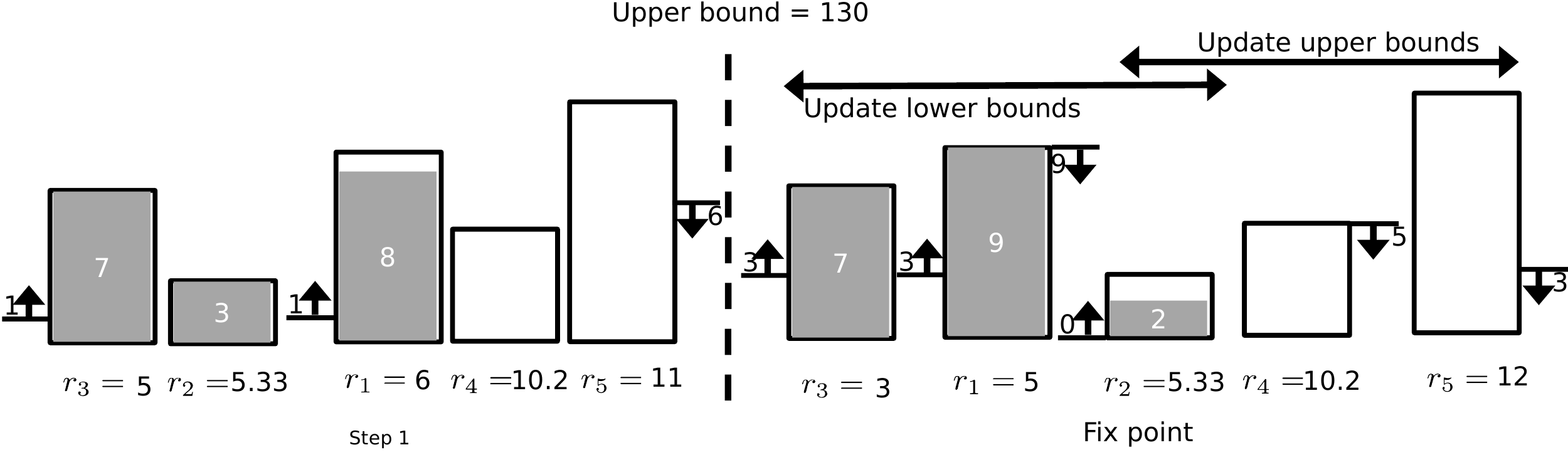} 
\caption{Propagation of lower and upper bounds of the load variables. 
The state of the domains at the end of the first iteration is depicted on the left.
   The lower bounds of $l_1$ and $l_3$ are increased to 1 and 
 the upper bound of $l_5$ is decreased to $6$.
  The state of the domains when the fix point is reached is shown on the right.}
\label{example3}
\end{figure}

Let's now add the propagation using dynamic programming ($\overline{l_3}$ would typically be reduce to 5). The lower bound observed at the root node is 119.66 whereas the one obtained by linear programming \emph{i.e.} with formulation~(\ref{lp-model}) (including the initial tightening of the capacities by dynamic programming) is 114.4.
\end{example}

\noindent\textbf{Algorithms and Complexity}. Assuming that $B^{'}$ and $W^{'}$ are available, $Lb_1(W^{'},B^{'})$ can be computed in $O(m\,log(m))$ time. Firstly we compute the $r_{j}$ values corresponding to  $B^{'}$ for all bins. Secondly, we sort the bins in non-decreasing $r_j$. Finally, the bound is computed by iterating over the sorted bins and the complexity is dominated by the sorting step. After computing $Lb_1(W^{'},B^{'})$, the values $a_j$ (the permutation of the bins) such that $r_{a_1} \leq r_{a_2} \leq \ldots \leq r_{a_m}$ are available as well as the critical $k$ and $L_k = W^{'} - \sum^{k-1}_{j = 1} C^{'}$. The update of $\underline{l_{a_j}}$ and  $\overline{l_{a_j}}$ can then be done in $O(m)$ as shown in Figure~\ref{propalgo}. Notice that this process can be repeated until a fix point is reached.

Algorithm \texttt{UpdateMinimumLoad} of Figure~\ref{propalgo} is used to update  the lower bounds of the load variables 
for each bin $a_j$  such that $r_{a_j} \leq r_{a_k}$. Recall that $k$ is the number of bins supporting the lower bound $Lb_1$.
If $j < k$ (resp.\ $j=k$) then the algorithm tries to find  the largest quantity, denoted by $q^-_{a_j}$, that can be removed
from the bin $a_j$ and can be put on the bins $a_b$, where $b \geq k$ (resp.\ $b \geq k+1$),
such that the increment in the cost remains smaller than the $gap$ (Lines~4--10). 
The algorithm loops over the bins as long as the  $q^-_{a_j}$ has not reached its maximum amount on bin $a_j$ \emph{i.e.} $L_{j}$ and the gap has not been overloaded.
When the increment in cost exceeds the gap then the minimum load if known.
Similarly, \texttt{UpdateMinimumLoad} of Figure~\ref{propalgo} is used to update  the upper bounds of the load variables 
for each bin $a_j$  such that $r_{a_j} \geq r_{a_k}$.

\begin{figure}[t]
\begin{small}
\centering

  \begin{minipage}{.45\textwidth}
  \centering

\begin{coden}
{ Algorithm 1: UpdateMinimumLoad} \\
  {\bf Input}: $a_{j}$ with $j \leq k$, $B^{'}$, $gap$ \\
  {\bf Output}: a lower bound of $l_{a_j}$ \\
  ~1.  $costInc \gets 0$; $q^-_{a_j} \gets 0$; $b \gets k$; \\
  ~2.  {\bf If} ($j = k$) \{$b \gets k+1$;\} \\
  ~3.  {\bf While} ($q^-_{a_j} < L_{j}$ \&\& $b \leq m$) \\
  ~4.   \> \> $loadAdd    \gets \min(L_{j}-q^-_{a_j}, C^{'}_{a_{b}} - L_b)$;\\
  ~5.   \> \> $costIncb     \gets loadAdd \times (r_{a_b}-r_{a_j})$;   \\ 
  ~6.   \> \> {\bf If} ($(costIncb + costInc) > gap$) \\ 
  ~7.   \> \> \> \> $q^-_{a_j} \gets q^-_{a_j} + \lfloor \frac{gap - costInc}{r_{a_b}-r_{a_j}} \rfloor$;\\
  ~8.   \> \> \> \> {\bf return} $\underline{l_{a_j}} + L_{j} - q^-_{a_j} $;\\
  ~9.   \> \> $costInc \gets \:\:costInc+costIncb$; \\ 
  ~10. \> \> $q^-_{a_j} \gets  \:\:q^-_{a_j}+loadAdd$;  $b = b+1$; \\
  ~11. {\bf return} $\underline{l_{a_j}}$ \\
  \end{coden}

\end{minipage}
\end{small}
  \begin{minipage}{0.45\textwidth}  
  \centering
  \vspace{-0.3cm}

\begin{small}
\begin{coden}
{Algorithm 2: UpdateMaximumLoad}\\
   {\bf Input}: $a_{j}$ with $j \geq k$ , $B^{'}$, $gap$\\
   {\bf Output}: an upper bound of $l_{a_j}$ \\
  ~1.  $costInc \gets 0$; $q^+_{a_j} \gets 0$; $b \gets k$; \\
  ~2.  {\bf If} ($j = k$) \{$q^+_{a_j} \gets L_k$; $b \gets k-1$;\} \\
  ~3.  {\bf While} ($q^+_{a_j} < C^{'}_{a_j}$ \&\& $b \geq 0$) \\
  ~4.   \> \> $loadAdd    \gets min(L_b, C^{'}_{a_j}-q^+_{a_j})$;   \\ 
  ~5.   \> \> $costIncb     \gets loadAdd \times (r_{a_j}-r_{a_b})$;   \\ 
  ~6.   \> \> {\bf If} $((costIncb + costInc) > gap)$ \\ 
  ~7.   \> \> \> \> $q^+_{a_j} \:\: \gets q^+_{a_j}+\lfloor \frac{gap - costInc}{r_{a_j}-r_{a_b}} \rfloor$;\\
  ~8.   \> \> \> \> {\bf return} $\underline{l_{a_j}} + q^+_{a_j}$;\\
  ~9.   \> \> $costInc\:\: \gets \:\:costInc+costIncb$; \\
  ~10. \> \> $q^+_{a_j}\:\: \gets \:\: q^+_{a_j}+loadAdd$; $b\: = b-1$; \\
  ~11. {\bf return} $\overline{l_{a_j}}$
  \end{coden}
\end{small}
\end{minipage}
\caption{Propagation algorithms for updating the lower and upper bounds of the load variables \label{propalgo}}
\end{figure}

\subsection{Dominance and symmetries.} The dominance and symmetries previously mentioned can be similarly eliminated here. We go a step further and eliminate some dominances during search. In particular, for two bins $i$, $j$ that are known to be open ($y_i = y_j = 1$)  such that $c_i \leq c_j$ and $C_i \geq C_j$, we can enforce $l_i \geq l_j$ in the remaining sub-tree. Regarding item symmetries, we prefer to avoid changing the domain's definitions (as done for Formulation~(\ref{lp-model})) to keep the semantics of the global constraint unchanged. Instead, each time a bin is proved to be infeasible for an item, the same value is removed from all the ungrounded items of the same size.
 
\subsection{Search.} The propagation scheme relies on the knowledge of good upper bounds so the search process is driven by this need. We branch first on the bin variables $y$ before the item's variables $x$. We branch on the cheapest bins first by selecting the $y$ with minimum $r$ value and setting it to 1. The $y$ variables are therefore sorted lexicographically by non-decreasing $r_j$ values. Once all $y$ variables are grounded, we select the open bin $k$ with the smallest slope ($c_k$) and assign to it the largest item that can fit in a perfect packing of the bin. We thus check by dynamic programming that the item selected belongs to a subset of items that perfectly (up to its capacity) fill the bin. The search is binary, \emph{i.e.} on the left branch we enforce $x_i = k$ and on the right branch we simply have $x_i \neq k$.

\subsection{Propagating a stronger lower bound.} The lower bound $z^{*}_2$ obtained by solving the linear relaxation of the cutting stock model can also be propagated by the constraint. However current restrictions of the domains have to be taken into account when computing the bound during search. The master and pricing problem are affected as we apply the bound by taking into account items already assigned (the capacity of the bins is reduced accordingly), and bins known to be open/close and currently possible items for each bin. The bound thus benefits from all other reasoning that acts on the domains. Since the pricing problem is time-consuming we found that the following techniques are important for efficiency:
\begin{itemize}
\item We keep the columns of the previous call to the cutting stock model that are still compatible (feasible) with respect to the current domains.
\item When solving the pricing problem for a bin, an upper bound is first computed by sorting the items from the most beneficial to the least and adding them greedily as long as the capacity is not overloaded. If the column obtained has a negative reduced cost, the dynamic programming algorithm is not called and the algorithm moves on. This is a very common technique in column generation that pays off when the pricing is costly and the master very easy to solve (because it usually increases the number of iterations). The dynamic programming algorithm is called at least once to prove that no negative reduced cost columns still exist and ensure the validity of the bound.
\end{itemize}